\documentclass[letterpaper, 10 pt, conference]{ieeeconf}
\IEEEoverridecommandlockouts 
\usepackage{amsmath}
\usepackage{color}
\usepackage{graphicx}
\usepackage{amsfonts}
\usepackage[noadjust]{cite} 
\newtheorem{remark}{Remark}
\newtheorem{theorem}{Theorem}
\newtheorem{corollary}{Corollary}

\IEEEoverridecommandlockouts

\def\bfu{{\mbox{\boldmath $u$}}}

\def\bfx{{\mbox{\boldmath $x$}}}

\newtheorem{assumption}{Assumption}

\title{\LARGE \bf Learning Model Predictive Control for Periodic Repetitive Tasks}
\author{Nicola Scianca$^1$\thanks{$^{1}$Nicola Scianca is with the Dipartimento di Ingegneria Informatica, Automatica e Gestionale, Sapienza Universit\`a di Roma, via Ariosto 25, 00185 Roma, Italy. E-mail: scianca@diag.uniroma1.it}, Ugo Rosolia$^2$, Francesco Borrelli$^2$\thanks{
$^{2}$Ugo Rosolia and Francesco Borrelli are with the Department of Mechanical Engineering, University of California at Berkeley, Berkeley, CA 94701, USA. E-mail: {\{ugo.rosolia, fborrelli\}}@berkeley.edu}}
\begin{document}

\maketitle

\begin{abstract}
We propose a reference-free learning model predictive controller for periodic repetitive tasks. We consider a problem in which dynamics, constraints and stage cost are periodically time-varying. The controller uses the closed-loop data to construct a time-varying terminal set and a time-varying terminal cost. We show that the proposed strategy in closed-loop with linear and nonlinear systems guarantees recursive constraints satisfaction, non-increasing open-loop cost, and that the open-loop and closed-loop cost are the same at convergence. Simulations are presented for different repetitive tasks, both for linear and nonlinear systems.
\end{abstract}

\section{Introduction}

Repetitive tasks can be found in virtually every human activity. Performing the same action several times can be tiring and alienating. For this reason in several applications, repetitive tasks are often delegated to automated systems. 
This motivated researchers to study how to design control algorithms tailored to iterative and repetitive tasks~\cite{bristow2006survey,c7,c8, hillerstrom1996repetitive,lee2001model,gupta2006period}. The key idea is to exploit historical data to improve the closed-loop performance of autonomous systems.
Indeed, even a small improvement, when applied to a large number of repetitions, will yield a considerable gain.


At each execution of an \textit{iterative task} the system starts from the same initial condition. These control problems are studied in  
Iterative Learning Control (ILC), where the controller learns from previous iterations how to improve its closed-loop performance. In classical ILC the controller objective is to track a given reference, rejecting periodic disturbances \cite{bristow2006survey,c7,c8}. The main advantage is that information from previous iterations is incorporated in the problem formulation and it is used to improve the system performance.

A different class of control problems arises when the controller performs a \textit{repetitive task} in which the initial condition of each execution is the final condition of the previous, which is to say, the system operates continuously.
These problems are studied in Repetitive Control (RC) \cite{hillerstrom1996repetitive}. The goal of RC is defined as tracking of a periodic trajectory, or rejection of a periodic disturbance. The idea is to construct a controller which contains a system whose output is the reference to be tracked (or the disturbance to be rejected), which is known as the internal model principle. Since RC aims at continuous operation, improvement is done by using previous data within a single execution of the task. Early RC controllers use frequency-domain design techniques, but there are modern formulations which employ Model Predictive Control (MPC) to track a given reference, such as \cite{lee2001model, gupta2006period, gondhalekar2011mpc}, or \cite{cao2008repetitive} which addresses the problem of period mismatch. Applications can be found for control of a reverse flow reactor \cite{balaji2007repetitive} or wind turbines \cite{friis2011repetitive}.


Classical ILC and RC approaches define performance in terms of tracking of a given reference, which has to be provided to the controller. The reference could be computed offline as a solution of a periodic optimal control problem \cite{gilbert1977optimal, bittanti1986optimal, colonius2006optimal}. Optimal periodic control is formulated as an infinite horizon control problem with the objective of minimizing an average cost. 

A reference trajectory may be hard to compute apriori, for this reasons several recent works proposed reference-free strategies~\cite{ellis2014tutorial,diehl2010lyapunov,zanon2016periodic}. 
Examples of these are provided by MPC schemes which minimize a certain economic index, often referred to as economic MPC~\cite{ellis2014tutorial}. These controllers compute the control action after forecasting the evolution of the system over a time horizon. This strategy often works well in practice but it makes it harder to provide theoretical guarantees. Among them, some specifically adopt a periodic formulation. A Lyapunov function for a class of systems is found for a periodic MPC with economic cost in \cite{diehl2010lyapunov}, while in \cite{zanon2016periodic} the authors analyze economic MPC on systems with dissipativity properties, including a periodic formulation.

In this work we extend the reference-free Learning Model Predictive Control (LMPC) described in~\cite{rosolia2017learning,rosolia2017linear} to repetitive tasks, where the system operates continuously and the initial condition is not the same at each execution. We maintain the core idea which is to use historical data to compute a terminal constraint and estimate the terminal cost of the MPC, however, the way these are computed is substantially different. The original formulation of LMPC employed the ILC paradigm in which the system is restarted to the same initial condition and uses only data from previous iterations. 
In the proposed work we continuously update the controller using data from a single execution. We consider dynamics, constraints and stage cost that are periodically time-varying and we require that an initial feasible periodic trajectory is known. For these systems we show recursive feasibility, non-increasing open-loop cost, and that the open-loop trajectory equals the closed-loop at convergence.

Several periodic MPC approaches make use of a \emph{lifted reformulation}~\cite{bittanti1986optimal} to turn a time-varying periodic system into a time-invariant system. This simplifies the problem if the prediction is an integer multiple than the period. We cannot use this approach in the proposed work since we assume the prediction horizon to be much smaller than the period. 


\section{Problem definition}\label{sec:problemDefinition}

Consider the time-varying periodic system
\begin{equation}
\begin{split}
\label{eq:systemDynamics}
&x_{t+1} = f_t(x_t,u_t) = f_{t+P}(x_t,u_t), \forall t \geq 0\\
\end{split}
\end{equation}
where $P$ is the period of the system\footnote{We assume the period P to be an integer multiple of the sampling time.}, the state $x_t \in \mathbb{R}^n$ and input $u_t \in \mathbb{R}^d$.
Furthermore, the system is subject to the following \mbox{$P$-periodic} state and input convex time-varying constraints,
\begin{equation}\label{eq:constraints}
\begin{split}
&x_t\in {\cal X}_t ={\cal X}_{t+P} \text{ and } u_t\in {\cal U}_t={\cal U}_{t+P}, \forall t \geq 0.
\end{split}
\end{equation}
Finally, we define the $j$th \emph{cycle} as the time interval $[jP,(j+1)P)$, with $j\ge 0$. At any time $t$, the current cycle is given by $M = \mathrm{floor}(t/P)$ and the time from the beginning of the cycle $\tau = t \bmod P$ is the \emph{intracycle time}.


The goal of the control problem is to find a \mbox{$P$-periodic} trajectory which solves the following average cost optimal control problem (see, e.g., \cite{diehl2010lyapunov})
\begin{equation}\begin{split}\label{eq:optimalControlProblem}
\min_{u_{0},u_1\ldots}
\quad & \lim_{T\to\infty}\frac{1}{T}\sum_{k=0}^{T-1} h_k(x_k,u_k) \\
\textrm{s.t.} \quad  &x_{t+1} = f_t(x_t,u_t),~  \forall t \geq 0\\
&x_t\in {\cal X}_t, u_t\in {\cal U}_t, ~ \forall t\geq 0\\
&x_t = x_{t+P}, ~ \forall t\geq 0.
\end{split}\end{equation}
where $h_t(\cdot,\cdot)$ is a convex function of the state and input, and it is time-varying with period $P$.
\begin{assumption}\label{ass:stageCost}
The constraints ${\cal X}_t$ and ${\cal U}_t$ are convex and $P$-periodic, i.e.,
${\cal X}_t ={\cal X}_{t+P} \text{ and } {\cal U}_t={\cal U}_{t+P},$
and the stage cost $h(\cdot, \cdot)$ is jointly convex in its arguments and also $P$-periodic
$h_t(\cdot,\cdot) = h_{t+P}(\cdot,\cdot) ~ \forall t\geq0$.
\end{assumption}
\smallskip

The control problem~\eqref{eq:optimalControlProblem} aims to find a periodic trajectory which minimizes an average cost. Note that this is in fact an infinite horizon control problem, but given the periodicity in the solution it is equivalent to minimizing the average cost over a single period~\cite{colonius2006optimal}.


\section{Controller Design}\label{sec:controllerdesign}
In this section we first show how to construct a terminal set and terminal cost function exploiting historical data. At time $t$, this data consists of the stored closed-loop trajectories and input sequence,
\begin{equation}\label{eq:storedTrajectories}
\begin{aligned}
 {\bf{u}}_t = [u_0,\ldots,~u_{t-1}] \text{ and } {\bf{x}}_t = [x_0,\ldots,~x_{t}].
\end{aligned}
\end{equation}
Afterwards, we use the terminal constraint set and cost function in the LMPC design.

\subsection{Safe Set}\label{sec:safeSet}
We show that the closed-loop data can be used to construct a time-varying terminal constraint set that allows us to guarantee recursive constraint satisfaction for the proposed strategy.
At each time $t$ we define the sampled safe set as
\begin{equation*}
{\cal SS}_t = \Bigg\{ \bigcup_{j=1}^M x_{t-jP} \Bigg\},
\end{equation*}
where $P$ is the period of system \eqref{eq:systemDynamics}. The above set contains all the states of the realized trajectory with the same intracycle time $\tau$ (defined in Section~\ref{sec:problemDefinition}). The rationale behind this choice is that the optimal problem (\ref{eq:optimalControlProblem}) is invariant for time shifts of any integer multiple of $P$.

We can now define the time-varying convex safe set ${\cal CS}_t$ as the convex hull of ${\cal SS}_t$, i.e.,
\begin{equation}\label{eq:convexSafeSet}
\begin{aligned}
{\cal CS}_t &= \text{Conv}\big( \mathcal{SS}_t \big) \\
&= \Bigg\{ x \in \mathbb{R}^n : \exists \lambda_j\ge 0, \sum_{i=1}^{M}\lambda_j = 1, x =\sum_{i=1}^M\lambda_i x_{t-jP}\Bigg\}.
\end{aligned}
\end{equation}

Figure~\ref{fig:safeSets} shows an example trajectory with $n=2$ in phase space. The associated safe sets are highlighted as light grey polygons. The period is $P=6$ which means there are a total of $6$ convex safe sets at any given time, one for every time within a cycle.

\begin{figure}[h!]
    \centering
    \includegraphics[scale=1.5]{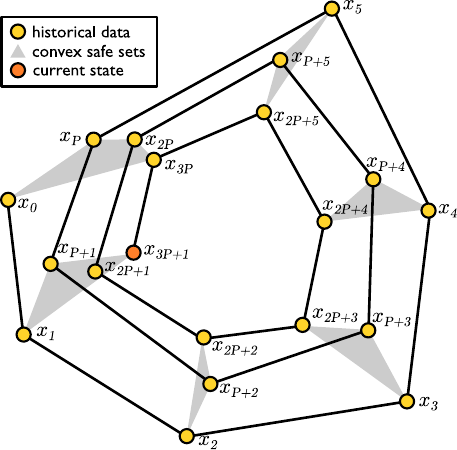}
    \caption{A visual representation of the convex safe sets for a period $P=6$.}
    \label{fig:safeSets}
\end{figure}

The convex safe set will be used as a terminal constraint for the LMPC scheme. Throughout the paper we make the following assumption, which guarantees that the safe set is non-empty. 
\begin{assumption}\label{ass:nonEmptySS}
We are given a periodic trajectory and associated input sequence $[x_0,\ldots, x_P] \text{ and }[u_0,\ldots, u_{P-1}]$
such that state and input constraint are satisfied $x_t \in \mathcal{X}_t$, $u_t \in \mathcal{U}_t ~ \forall k \in \{0, \ldots, P \}$ and $x_P = x_0$.
\end{assumption}
\medskip

\begin{figure*}
    \centering
    \includegraphics[scale=1.5]{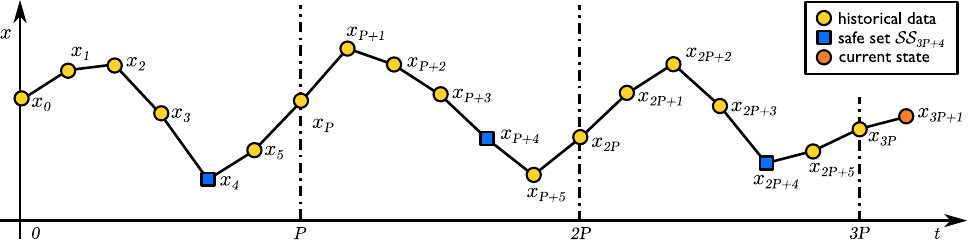}
    \caption{A visual representation of how a safe set is constructed. In the example the period is $P=6$ and the prediction is $N=3$. Yellow points show the historical data. Among them, some are selected (shown in blue) as part of the safe set ${\cal SS}_{3P+4}$, which is used to construct the terminal set of the MPC.}
    \label{fig:stateTime}
\end{figure*}

\begin{remark}
The above assumption is not restrictive in most practical applications. Indeed, a periodic suboptimal trajectory can be computed using any other controller or it can simply be a known feasible trajectory. For example, any steady-state solution to the dynamics can be used as long as it satisfies the constraints.
\end{remark}

\subsection{\mbox{$Q$-function}}
In LMPC for iterative tasks for each stored state we compute a cost-to-go, which is defined as the cost to complete the control task~\cite{rosolia2017learning}. However, our system operates continuously and our goal is achieved when the system reaches an optimal periodic trajectory, which means that the cost-to-go is not well defined. For this reason, in place of the cost-to-go, we associate to every state a return-cost, which is defined as the cost to return to the current state $x_t$. 
This return-cost is then used to define a \mbox{$Q$-function} which will be employed in the controller design.



We start by associating a return-cost to each state in the safe set $x_i\in {\cal SS}_t$, where $i$ denotes the time when this state was realized in the trajectory. At time $t$, the return-cost is defined as the cost to reach the current state $x_t$ starting from $x_i$ and following the realized trajectory, i.e.,
\begin{equation}
\label{eq:ctg}
J_t(x_i) = \sum_{k=i}^{t-1}h_k(x_k,u_k).
\end{equation}
It is easy to verify that the following property holds
\begin{equation}\label{eq:ctgProp}
J_t(x_i) - h_i(x_i,u_i) + h_t(x_t,u_t) = J_{t+1}(x_{i+1}),
\end{equation}
which gives the return-cost of the states in the safe set ${\cal SS}_{t+1}$.

\begin{remark}
The return-cost is time-varying and therefore has to be computed at each time $t$. However, the property~\eqref{eq:ctgProp} allows for efficiently computing $J_t(\cdot)$ by updating $J_{t-1}(\cdot)$.
\end{remark}

The \mbox{$Q$-function} is defined for every state $x\in {\cal CS}_k$ as
\begin{equation}
\label{eq:Qfunction}
\begin{aligned}
Q_{k\to t}\left(x\right) = \min_{\lambda_j \geq 0} \quad &\sum_{j=1}^M\lambda_j J_t\left( x_{k-jP}\right)\\
\text{s.t.} \quad & \sum_{j=1}^M\lambda_j x_{k-jP} =x,\quad\sum_{j=1}^M\lambda_j = 1.
\end{aligned}
\end{equation}
For an intuitive interpretation, consider that by construction every state $x\in{\cal CS}_t$ can be written as a convex combination of states $x_j\in{\cal SS}_t$. We define $Q_t(x)$ as the same convex combination of the return-cost of the states in ${\cal SS}_t$.

\subsection{LMPC Formulation}

In this section we describe the proposed LMPC strategy for periodic repetitive tasks.
The core idea is to use the convex safe set ${\cal CS}_t$ and the \mbox{$Q$-function} as a terminal constraint and cost for the MPC scheme. 
The difference between the proposed LMPC and the original scheme~\cite{rosolia2017learning} is that: $(i)$ the convex safe set is defined based on the periodicity of the control problem, $(ii)$ the \mbox{$Q$-function} is defined as the cost to return to the current state, $(iii)$ dynamics, constraints and stage cost are allowed to be periodically time-varying.

A step of the LMPC algorithm goes as follows: at a generic time $t$ we compute the solution to the following finite time optimal control problem with horizon $N$,
\begin{equation}\label{eq:FTOCP}
\begin{split}
J_{t\to t+N}^{\rm LMPC}(x_t) =& \\
\min\limits_{\substack{u_{t|t},\dots\\u_{t+N-1|t}}} \quad &\!\!\!\!\sum\limits_{k=t}^{t+N-1}\!\!\! h_k(x_{k|t},u_{k|t}) + {Q}_{t+N\to t}(x_{t+N|t}) \\
\textrm{s.t.} \quad
&x_(k+1|t) = f(x_{k|t}, u_{k|t}) \quad \forall k\in[0,N-1]\\
&x_{t|t} = x_t\\
&x_{t+N|t} \in {\cal CS}_{t+N} \\
&x_{k|t}\in {\cal X}_k, u_{k|t}\in {\cal U}_k \quad \forall k\in[0,N-1].
\end{split}
\end{equation}
The solution to the above finite time optimal control problems steers the system from the current measured state $x_t$ to the convex safe set ${\cal CS}_{t+N}$, while satisfying state and input constraint.

Let
\begin{equation*}\begin{split}
\bfu_{t:t+N}^\ast &= [ u_{t|t}^\ast,\dots,u_{t+N-1|t}^\ast ] \\
\bfx_{t:t+N}^\ast &= [ x_{t|t}^\ast,\dots,x_{t+N|t}^\ast ]
\end{split}\end{equation*}
be the optimal solution to~\eqref{eq:FTOCP}, then we apply to the system~\eqref{eq:systemDynamics} the first element of the optimizer vector
\begin{equation}\label{eq:LMPCpolicy}
    u_t = u_{t|t}^\ast.
\end{equation}
The process is repeated at time $t+1$ starting from $x_{t+1}$, which is the standard MPC procedure.

Figure~\ref{fig:stateTime} shows, for a given realized trajectory, how the safe set is constructed. In the example the state dimension is $n=1$, the period is $P=6$ and the prediction horizon is $N=3$. The current time is $t=3P+1$, so the safe set $\mathcal{SS}_{3P+4} = \{x_{1+N},x_{P+1+N},x_{2P+1+N}\}$.

Figure~\ref{fig:stateSpace} gives another visual representation of a generic step of the algorithm. The parameters are the same except for the state dimension $n=2$, and the trajectory is shown in phase space. The realized trajectory starts at $x_0$ and roughly delineates a spiral shape. As in the previous example, the current time is $t=3P+1$ and we select the same states for the safe set. Here the convex safe set ${\cal CS}_{3P+4}$ is also shown as a grey polygon with blue vertices.

\begin{figure}[h!]
    \centering
    \includegraphics[scale=1.5]{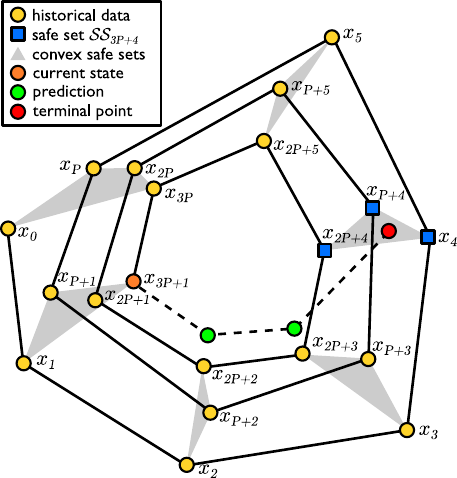}
    \caption{A visual representation of a step of the MPC algorithm. The grey regions show convex safe sets. The safe set currently used to construct the terminal constraint is marked by blue vertices.}
    \label{fig:stateSpace}
\end{figure}

It is important to underline that only the terminal state is required to be in the convex safe set, not the entire prediction, as it can be observed by the fact that the green states are outside of any of the safe sets. Therefore, the states added each time to the realized trajectory do not always lie in the existing safe set, but rather enlarge the safe set at the next cycle. This is a crucial point because it is what drives the learning process.

\section{Properties}\label{sec:properties}
In this section, we exploit the properties of the convex safe set and \mbox{$Q$-function} from Section~\ref{sec:controllerdesign} to show that the proposed control strategy guarantees recursive constraint satisfaction and non-increasing open-loop cost. For linear systems, these results are direct consequences of the return-cost property~(\ref{eq:ctgProp}). In order to extend these results to a class of nonlinear systems, we adopt two additional assumptions:

\begin{assumption}\label{eq:dynamicsAssumption}
Consider a set of state-input pairs $(x_j, u_j)$
\begin{equation*}
\begin{split}
&\{(x_j,u_j)\quad j=0,\dots,L\} \\
&x_j\in {\cal X}_t, u_j\in {\cal U}_t \quad \forall j=0,\dots,L\\
\end{split}
\end{equation*}
such that $f_t(x_j,u_j) \in {\cal X}_{t+1}$ $\forall j=0,\dots,L$. We assume that for any set of multipliers $\lambda_j\ge 0$ such that $\sum_j \lambda_j = 1$, there exists a set of multipliers $\gamma_j>0$ such that $\sum_j \gamma_j = 1$ and 
\begin{equation*}
f_t\left( \sum_{j=0}^L \lambda_j x_j, \sum_{j=0}^L \gamma_j u_j\right) = \sum_{j=0}^L \lambda_j f(x_j, u_j), ~  \forall t\ge 0.
\end{equation*}
\end{assumption}
\medskip
Note that this assumption is verified for linear systems using $\gamma_j = \lambda_j$.

\begin{assumption}\label{eq:costAssumption}
If the system dynamics~(\ref{eq:systemDynamics}) is nonlinear, the stage cost does not depend on the input.
\end{assumption}

\subsection{Recursive feasibility}

\begin{theorem}[Recursive Feasibility]
Consider system~\eqref{eq:systemDynamics} in closed-loop with the LMPC~\eqref{eq:FTOCP} and \eqref{eq:LMPCpolicy}, where the convex safe set and \mbox{$Q$-function} are defined in \eqref{eq:convexSafeSet} and \eqref{eq:Qfunction}, respectively. Let Assumptions~\ref{ass:stageCost}-\ref{eq:dynamicsAssumption} hold. Then at all time $t\geq P$, Problem~\eqref{eq:FTOCP} is feasible, i.e., the closed-loop system~\eqref{eq:systemDynamics} and \eqref{eq:LMPCpolicy} satisfies input and state constraint~\eqref{eq:constraints}.
\end{theorem}

\begin{proof}
By Assumption~\ref{ass:nonEmptySS}, a feasible trajectory for $t \in \{0, \ldots, P\}$ is given, thus the system is in closed-loop with LMPC starting from time $t=P$. First we notice that at time $t = P$, the following state and input sequence
\begin{equation*}
    [x_0,\ldots,x_P] \text{ and } [u_0,\ldots,u_{P-1}]
\end{equation*}
from Assumption~\eqref{ass:nonEmptySS} is feasible for Problem~\eqref{eq:FTOCP}, as this trajectory is periodic. Therefore, at time $t=P$ the LMPC is feasible.

We will now prove that if LMPC is feasible at time $t$, it is also feasible at time $t+1$, for any $t$. Consider the optimal predicted trajectory at time $t$, where the terminal point is in ${\cal CS}_{t+N}$,
\begin{equation*}
\begin{split}
\left[x_{t|t}^\ast,\dots,x_{t+N|t}^\ast=\sum_{j=1}^M \lambda_j^\ast x_{t+N-jP}\right].
\end{split}
\end{equation*}


Assumption~\ref{eq:dynamicsAssumption} guarantees that there exists an input $\sum_{j=1}^M \gamma_j u_j$ such that
\begin{equation*}\begin{split}
&f_t\left( \sum_{j=1}^M \lambda_j^\ast x_{t+N-jP}, \sum_{j=1}^M \gamma_j u_{t+N-jP}\right) = \\
&\sum_{j=1}^M \lambda_j^\ast f_t(x_{t+N-jP}, u_{t+N-jP}) = \sum_{j=1}^M \lambda_j^\ast x_{t+1+N-jP}.
\end{split}\end{equation*}
This is a feasible input because $\sum_{j=1}^M\gamma_j u_{t+N-jP} \in {\cal U}_{t+N}$ and the new state is a feasible terminal state at time $t+1$ as $\sum_{j=1}^M \lambda_j^\ast x_{t+1+N-jP}\in {\cal CS}_{t+1}$. This implies that the following
\begin{equation}\label{eq:feasTraj}
\left[x_{t+1|t}^\ast,\dots,x_{t+N|t}^\ast,\sum_{j=1}^M \lambda_j x_{t+1+N-jP}\right]
\end{equation}
is a feasible trajectory for LMPC at time $t+1$.

We have shown the the LMPC is feasible at time $t=P$. Furthermore, we have that if the LMPC is feasible at time $t$, then the LMPC is feasible at time $t+1$. Therefore, we conclude by induction that the LMPC~\eqref{eq:FTOCP} and \eqref{eq:LMPCpolicy} is feasible for all $t \geq P$ and that the closed-loop system~\eqref{eq:systemDynamics} and \eqref{eq:LMPCpolicy} satisfies input and state constraint~\eqref{eq:constraints}.
\end{proof}

\subsection{Non-increasing cost}

The following theorem proves non-increasing cost of LMPC for linear system. The extension to the nonlinear case is given in a subsequent corollary.

\begin{theorem}[Non-increasing cost, LTV systems]\label{theo:nonincreasingCostNL}
Consider system~\eqref{eq:systemDynamics} in closed-loop with the LMPC~\eqref{eq:FTOCP} and \eqref{eq:LMPCpolicy}, where the convex safe set and \mbox{$Q$-function} are defined in \eqref{eq:convexSafeSet} and \eqref{eq:Qfunction}, respectively. Let system~\eqref{eq:systemDynamics} be linear time-varying and Assumptions~\ref{ass:stageCost}-\ref{eq:dynamicsAssumption} hold. Then for all time $t\geq P$, we have that
\begin{equation*}
J_{t\to t+N}^{\rm LMPC}(x_t) \geq J_{t+1\to t+1+N}^{\rm LMPC}(x_{t+1}).
\end{equation*}
\end{theorem}
\medskip
\begin{proof}
We will now prove that the cost of LMPC is non-increasing at each time instant. To do this we will employ property (\ref{eq:ctgProp}) of the return-cost.
Consider the cost of LMPC at time $t$
\begin{equation*}
\begin{split}
&J_{t\to t+N}^{LMPC}(x_t) = \sum\limits_{k=t}^{t+N-1}\!\!\! h_k(x_{k|t}^\ast,u_{k|t}^\ast) + Q_{t+N\to t}(x_{t+N|t}^\ast) \\
&= \sum_{k=t}^{t+N-1}\!\!\! h_k(x_{k|t}^\ast,u_{k|t}^\ast) + \sum_{j=1}^M\lambda_j^\ast J_t( x_{t+N-jP}) \\
& =\sum_{k=t}^{t+N-1}\!\!\! h_k(x_{k|t}^\ast,u_{k|t}^\ast) +\sum_{j=1}^M\lambda_j^\ast J_{t+1}( x_{t+1+N-jP}) \\
& +\sum_{j=1}^M\lambda_j^\ast h_{t+N-jP}( x_{t+N-jP}, u_{t+N-jP}) - \sum_{j=1}^M\lambda_j^\ast h_t(x_t,u_t)
\end{split}
\end{equation*}
where we used the property of the return-cost~(\ref{eq:ctgProp}).
In the above equation, the stage costs $h_{t+N-jP}(\cdot,\cdot)$ have the same intracycle time, thus from Assumption \ref{ass:stageCost} they can be replaced with $h_{t+N}(\cdot,\cdot)$. Furthermore, we can isolate the initial term in the first sum
\begin{equation*}
\begin{split}
&J_{t\to t+N}^{LMPC}(x_t) = \\
&h_t(x_t,u_t) + \!\!\!\sum_{k=t+1}^{t+N-1} \!\!\! h_k(x_{k|t}^\ast,u_{k|t}^\ast) +
\sum_{j=1}^M\lambda_j^\ast J_{t+1}( x_{t+N-jP+1})\\
&+\sum_{j=1}^M\lambda_j^\ast h_{t+N}( x_{t+N-jP}, u_{t+N-jP}) - \sum_{j=1}^M\lambda_j^\ast h_t(x_t,u_t).
\end{split}
\end{equation*}
Since the multipliers $\lambda_j$ add up to $1$, the last term in the above equation is $-h_t(x_t,u_t)$ and cancels the first term
\begin{equation*}\label{eq:lastEquality}
\begin{split}
&J_{t\to t+N}^{LMPC}(x_t) = \\
& \quad \sum_{k=t+1}^{t+N-1}\!\!\! h_k(x_{k|t}^\ast,u_{k|t}^\ast) + \sum_{j=1}^M\lambda_j^\ast h_{t+N}( x_{t+N-jP}, u_{t+N-jP})\\
& \quad \quad \quad \quad \quad \quad\quad \quad\quad \quad \quad \quad \quad +\sum_{j=1}^M\lambda_j^\ast J_{t+1}( x_{t+N-jP+1}).
\end{split}
\end{equation*}
By recalling that the stage cost is convex we have

\begin{equation*}
\begin{split}
J_{t\to t+N}^{LMPC}(x_t) &\geq
\sum_{k=t+1}^{t+N-1} \!\!\!h_k(x_{k|t}^\ast,u_{k|t}^\ast)\\
&+h_{t+N}\left(\sum_{j=1}^M \lambda_j^\ast  x_{t+N-jP},\sum_{j=1}^M \lambda_j^\ast  u_{t+N-jP}\right) \\
&+\sum_{j=1}^M\lambda_j^\ast J_{t+1}( x_{t+N-jP+1}).
\end{split}
\end{equation*}
Notice from \eqref{eq:Qfunction} that the last sum in the above inequality is an upper bound of the \mbox{$Q$-function} at $t+1$ for the terminal point of trajectory \eqref{eq:feasTraj}, therefore
\begin{equation}\label{eq:feasibleTrajectoryCost}
\begin{split}
J_{t\to t+N}^{LMPC}(x_t) &\geq
\sum_{k=t+1}^{t+N-1} \!\!\!h_k(x_{k|t}^\ast,u_{k|t}^\ast)\\
&+h_{t+N}\left(\sum_{j=1}^M \lambda_j^\ast  x_{t+N-jP},\sum_{j=1}^M \lambda_j^\ast  u_{t+N-jP}\right) \\
&+Q_{t+1\to t+1+N}\left(\sum_{j=1}^M\lambda_j^\ast x_{t+N-jP+1}\right).
\end{split}
\end{equation}
By assumption the system is linear, thus the right-hand side of this inequality is the cost of the feasible trajectory~(\ref{eq:feasTraj}) (recall that for a linear system (\ref{eq:feasTraj}) is verified with $\gamma_j=\lambda_j$). This implies
\begin{equation*}
\begin{split}
&J_{t\to t+N}^{LMPC}(x_t) \geq J_{t+1\to t+N+1}^{LMPC}(x_{t+1})
\end{split}
\end{equation*}
which proves that the cost is non-increasing.
\end{proof}
\medskip
\begin{corollary}[Non-increasing cost, nonlinear systems]
Consider the nonlinear system~\eqref{eq:systemDynamics} in closed-loop with the LMPC~\eqref{eq:FTOCP} and \eqref{eq:LMPCpolicy}, where the convex safe set and \mbox{$Q$-function} are defined in \eqref{eq:convexSafeSet} and \eqref{eq:Qfunction}, respectively. Let Assumptions~\ref{ass:stageCost}-\ref{eq:costAssumption} hold. Then for all time $t\geq P$, we have that
\begin{equation*}
J_{t\to t+N}^{\rm LMPC}(x_t) \geq J_{t+1\to t+1+N}^{\rm LMPC}(x_{t+1}).
\end{equation*}
\end{corollary}
\medskip
\begin{proof}
All steps in the proof of Theorem~\ref{theo:nonincreasingCost} are still valid except for the conclusion that \eqref{eq:feasibleTrajectoryCost} is the cost of a feasible trajectory (because $\gamma_j\neq \lambda_j$ in \eqref{eq:feasTraj}). However recall that by Assumption~\ref{eq:costAssumption} we have that the stage cost does not depend on the input. Thus \eqref{eq:feasibleTrajectoryCost} becomes
\begin{equation}\label{eq:feasibleTrajectoryCostNL}
\begin{split}
J_{t\to t+N}^{LMPC}(x_t) &\geq
\!\!\!\!\sum_{k=t+1}^{t+N-1}\!\!\!\! h_k(x_{k|t}^\ast,u_{k|t}^\ast)\!+\!h_{t+N}\!\left(\sum_{j=1}^M \lambda_j^\ast  x_{t+N-jP}\right)\\
&
+Q_{t+1\to t+1+N}\left(\sum_{j=1}^M\lambda_j^\ast x_{t+N-jP+1}\right)
\end{split}
\end{equation}

where the right-hand side is the cost of the feasible trajectory~(\ref{eq:feasTraj}). This implies that
\begin{equation*}
\begin{split}
&J_{t\to t+N}^{LMPC}(x_t) \geq J_{t+1\to t+N+1}^{LMPC}(x_{t+1})
\end{split}
\end{equation*}
which proves that the cost is non-increasing.
\end{proof}

\begin{theorem}[Performance improvement, LTV systems]\label{theo:nonincreasingCost}
Consider system \eqref{eq:systemDynamics} in closed-loop with the LMPC~\eqref{eq:FTOCP} and \eqref{eq:LMPCpolicy}, where the convex safe set and \mbox{$Q$-function} are defined in \eqref{eq:convexSafeSet} and \eqref{eq:Qfunction}, respectively. Let system \eqref{eq:systemDynamics} be LTV and Assumptions~\ref{ass:stageCost}-\ref{eq:dynamicsAssumption} hold. Furthermore, assume that for $t>c$ the closed loop system converges to a \mbox{$P$-periodic} trajectory
\begin{equation*}
\bfx_c^\ast = [x_c^\ast,x_{c+1}^\ast,\dots,x_{c+P-1}^\ast],
\end{equation*}
and that the stage cost $h(\cdot,\cdot)$ is strictly convex.
Then the following closed-loop cost over a period $J_{t\to t+P}(x_t)$ is equal to the open-loop LMPC cost, i.e.,
\begin{equation*}
J_{t\to t+P}(x_t) = \sum_{k=t}^{t+P-1}h_k(x_k^\ast,u_{k}^\ast) = J_{t\to t+N}^{\rm LMPC}(x_{t})\quad \forall t>c.
\end{equation*}
\end{theorem}
\medskip
\begin{proof}
The proof will consist of two parts. The first will prove that, if the trajectory is periodic, the LMPC terminal point is on the closed-loop trajectory. The second part will show that, assuming strictly convex stage cost, the open-loop and closed-loop trajectories are identical.

Assume in the following $t\ge c$. 
Since the trajectory is \mbox{$P$-periodic} we have that the closed-loop cost over a period is constant, i.e.,
\begin{equation*}
J_c^\ast = \sum_{k=t}^{t+P-1}h_k(x_k^\ast,u_{k}^\ast).
\end{equation*}
From the definition of the return-cost \eqref{eq:ctg} we have
\begin{equation}\label{eq:recursiveCtg}
J_{t+P}(x_c) = J_t(x_c) + J_c^\ast.
\end{equation}
It is easy to see from \eqref{eq:ctg} that the return-cost of every state in the safe set increases by $J_c^\ast$ after one period. Furthermore, the state $x_{c+N}$ is re-added to the safe set at each cycle, so its cost stays constant. Thus we can characterize the variation of the \mbox{$Q$-function} between two consecutive cycles as
\begin{equation*}
Q_{k+P\to t}(x) = Q_{k\to t}(x) + \Lambda(x) \quad \forall k,t\ge c
\end{equation*}
with $0\ge\Lambda(x)\ge J_c^\ast$. Also note that $\Lambda(x) = 0$ if and only if $x = x_c^\ast$, i.e., the terminal state is on the trajectory $\bfx_c^\ast$.
Consider the cost of LMPC at time $t+P$
\begin{equation*}\begin{split}
&J_{t+P\to t+P+N}^{LMPC}(x_c) =\\ &\sum_{k=t+P}^{t+P+N-1}h_k(x_{k|t+P}^\ast,u_{k|t+P}^\ast) + Q_{t+P+N\to t}(x_{t+P+N|t}^\ast) \\
&=\sum_{k=t+P}^{t+P+N-1}h_k(x_{k|t+P}^\ast,u_{k|t+P}^\ast) + Q_{t+N\to t}(x_{t+P+N|t}^\ast)\\
&\quad\quad\quad\quad\quad\quad\quad\quad\quad\quad\quad\quad\quad\quad\quad\quad\quad\quad+\Lambda(x_{t+N|t}^\ast)
\end{split}\end{equation*}
$x_{t+P+N|t}^\ast$ is not necessarily the optimal terminal state for LMPC at time $t$, but it is feasible because the convex safe set does not change after convergence. Thus we can lower bound the above expression with $J_{t\to t+P}^{LMPC}(x_c)+\Lambda(x)$ and upper bound it using Theorem~\ref{theo:nonincreasingCost}
\begin{equation*}\begin{split}
&J_{t\to t+N}^{LMPC}(x_{t+N|t}^\ast) \ge J_{t+P\to t+P+N}^{LMPC}(x_{t+P+N|t}^\ast) \ge \\
&J_{t\to t+N}^{LMPC}(x_{t+N|t}^\ast) + \Lambda(x_{t+N|t}^\ast)
\end{split}\end{equation*}
which implies $\Lambda(x_{t+N|t}^\ast)=0$ and thus $x_{t+N|t}^\ast=x_{c+N}$. This completes the first part of the proof.

In the second part we will prove that the open-loop trajectory is the same as the closed loop trajectory $\forall t\ge c$. We proceed by contradiction by assuming that the open-loop trajectory is not the same as the closed-loop. Therefore, for some $t$, the predicted trajectory does not overlap with the predicted trajectory at time $t+1$. With no loss of generality we will assume this happens at time $c$.

From convergence it follows that the first two elements in the open-loop trajectory are on the closed-loop trajectory (one is the current state, the other is the realized state). We also proved that the terminal state is on the closed-loop trajectory
\begin{equation*}
\bfx_{c|c}^\ast = [x_{c|c}^\ast,x_{c+1|c}^\ast,\dots,x_{c+N|c}^\ast],
\end{equation*}
in which $x_{c|c}^\ast=x_c$, $x_{c+1|c}^\ast=x_{c+1}$ and $x_{c+N|t}^\ast=x_{c+N}$, and the predicted trajectory at time $c+1$
\begin{equation}\label{eq:trajectoryXC1}
\bfx_{c+1|c+1}^\ast = [x_{c+1|c+1}^\ast,x_{c+2|c+1}^\ast,\dots,x_{c+N+1|c+1}^\ast],
\end{equation}
in which $x_{c+1|c+1}^\ast=x_{c+1}$, $x_{c+2|c+1}^\ast=x_{c+2}$, and $x_{c+N+1|c+1}^\ast=x_{c+N+1}$.

Construct a new candidate trajectory by appending $x_{c+N+1}$ at the end of $\bfx_{c|c}^\ast$ and removing its first term,
\begin{equation}\label{eq:trajectoryXTT}
\bar \bfx_{c+1} = [x_{c+1|c}^\ast,\dots,x_{c+N|c}^\ast,x_{c+N+1}^\ast].
\end{equation}
By the same argument used in the proof for Theorem~\ref{theo:nonincreasingCost}, the cost of trajectory~\eqref{eq:trajectoryXC1} is greater or equal of the cost of trajectory~\eqref{eq:trajectoryXTT}, but since we have reached convergence the cost is in fact equal. Furthermore, these two trajectory also have the same initial and terminal states.
Problem~\eqref{eq:FTOCP} is strictly convex because, by assumption, $h(\cdot, \cdot)$ is strictly convex, the dynamics is LTV and the constraint set is convex. It follows then that the trajectory~\eqref{eq:trajectoryXTT} equals the trajectory~\eqref{eq:trajectoryXC1}. Therefore we have that
\begin{equation*}
    x_{c+2} = x^*_{c+2|c+1} = x^*_{c+2|c}.
\end{equation*}
By iterating the above procedure we have that 
\begin{equation*}
    x_{c'+k} = x^*_{c'+k|c'+1} = x^*_{c'+k|c'}, \forall k \in \{0,\dots, N\}, \forall c' \geq c.
\end{equation*}
Thus the closed-loop trajectory and the open-loop trajectory are the same.
\end{proof}

\section{Simulations}\label{sec:simulations}

We test the proposed LMPC on $4$ systems. 
In each simulation the period is $P=100$. The first three are performed on linear systems and each one focuses on a different element of the control problem enforcing the periodic behavior. In particular, we test the controller on examples with time-varying dynamics, time-varying constraints and time-varying stage cost. We consider them separately for illustrative purposes, but this is not required. Finally, we test the proposed strategy on a time-varying nonlinear system.

For each simulation, we show the state and input time evolution for ten periods.
Furthermore, we report the LMPC cost which is decreasing until the closed-loop system converges to a periodic trajectory.
Note that the first cycle shown in each figure is the initial known trajectory.

\subsection{Linear system with time-varying dynamics}

The first simulation is performed on the following linear time-varying system $x_{t+1} = A_t x_t + Bu_t$ where state $x=(p,q)^T$,

\begin{equation*}
\begin{split}
&A_t = \left(\begin{array}{cc}
1 & 0.1 \\
0.1(1-\sin (2\pi t/P)) & 1
\end{array}\right) \text{ and }
B = \left(\begin{array}{c}
0 \\
0.1
\end{array}\right).
\end{split}
\end{equation*}
The control objective is  to reach a set-point on the first component of the state $-0.2$, while minimizing the input, i.e.
$h(x_t,u_t) = (p_t-0.2)^2 + u_t^2 \quad \forall t\ge 0$.
Furthermore, the state constraint is $|p_t|\leq 0.3~\forall t\ge 0$. The initial feasible trajectory is a zero input steady-state at the origin and the prediction horizon is $N=25$.

\begin{figure}[h!]
\centering
\includegraphics[width=\columnwidth]{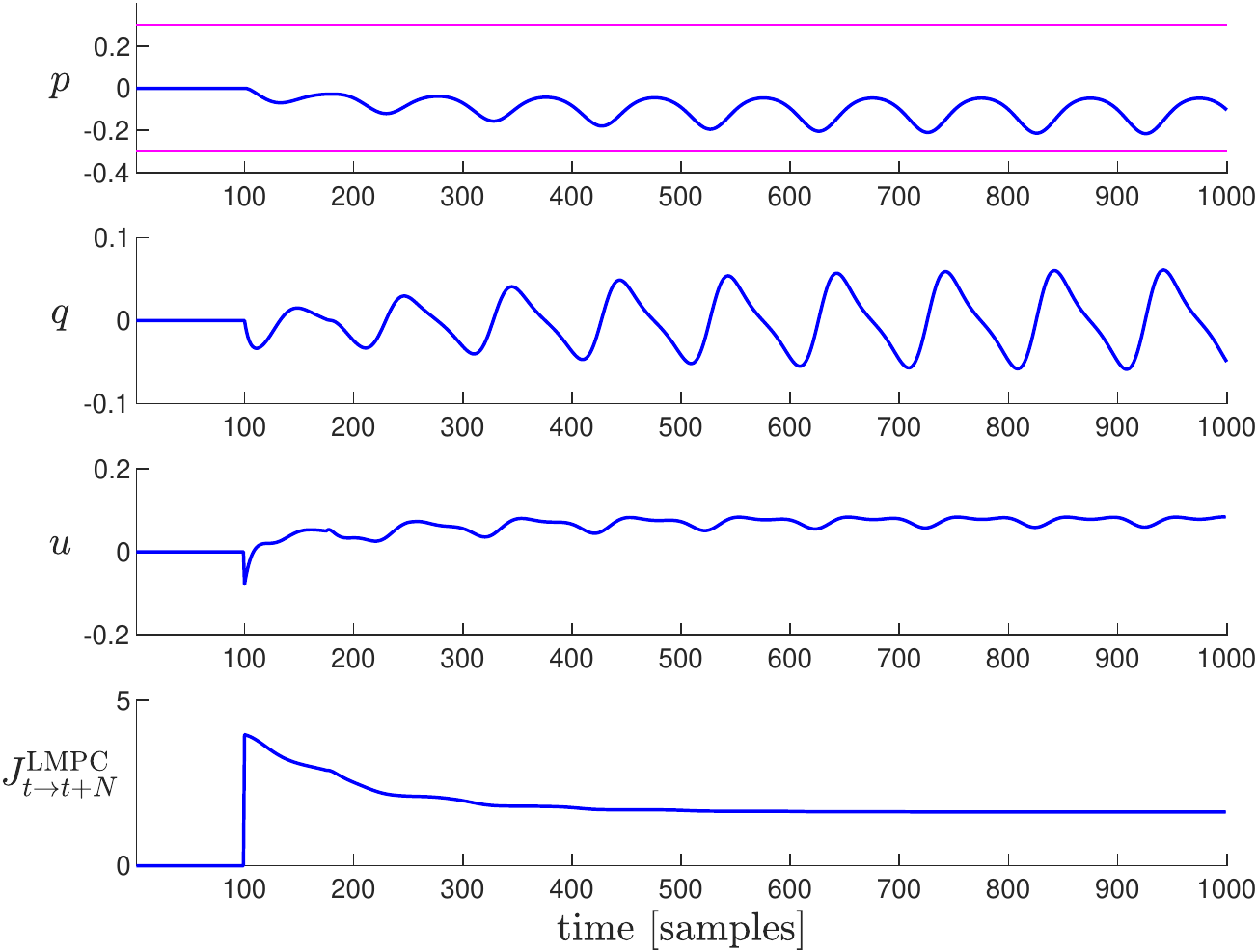}
\caption{Linear time-varying system controlled with LMPC.}
\label{fig:linearTVdynamics}
\end{figure}

The dynamics can be interpreted as a spring-mass system, where the spring stiffness is varying periodically. As such, we expect the first component of the state to settle somewhere in between the origin and the set-point, with the exact position depending on the value of the variable stiffness at that particular time. Figure~\ref{fig:linearTVdynamics} shows the trajectory over a time frame of ten cycles, and we confirm that the controller recursively satisfies the state constraints. Furthermore, we notice that the cost of LMPC is decreasing and after 3-4 cycles the system reaches a periodic trajectory which behaves as expected.

\subsection{Linear system with time-varying constraints}

The second simulation is performed on a double integrator subject to time-varying constraints. The state is $x=(p,q)^T$ and the dynamics are $x_{t+1} = Ax_t + Bu_t $ with
\begin{equation}
\begin{split}\label{eq:doubleIntegrator}
&A = \left(\begin{array}{cc}
1 & 0.1 \\
0 & 1
\end{array}\right) \quad
B = \left(\begin{array}{c}
0 \\
0.1
\end{array}\right).
\end{split}
\end{equation}
The prediction horizon is $N=30$ and the and the objective is the minimization of the input, i.e.
$h_t(x_t,u_t) = u_t^2 \quad \forall t$.
The time-varying constraints follow a cyclical a pattern that repeats every $P$.
\begin{equation*}
\begin{cases}
-0.4\le p\le 0.1 &\quad\text{if }t \bmod P < P/6\\
-0.4\le p\le -0.2 &\quad\text{if }P/6\ge t\bmod P< 2P/6 \\ 
-0.4\le p\le 0.1 &\quad\text{if }2P/6\ge t\bmod P< 3P/6 \\
-0.1\le p\le 0.4 &\quad\text{if }3P/6\ge t\bmod P< 4P/6 \\ 
0.2\le p\le 0.4 &\quad\text{if }4P/6\ge t\bmod P< 5P/6 \\ 
-0.1\le p\le 0.4 &\quad\text{if }t\bmod P\ge 5P/6 \\ 
\end{cases}
\end{equation*}

There is no steady-state solution which satisfies the constraints in this example. The initial trajectory was computed using an MPC which minimizes the input over the prediction horizon but has no terminal constraint or terminal cost. After a few iteration it converges to a periodic suboptimal trajectory which is then used to initialize LMPC.
\begin{figure}[h!]
\centering
\includegraphics[width=\columnwidth]{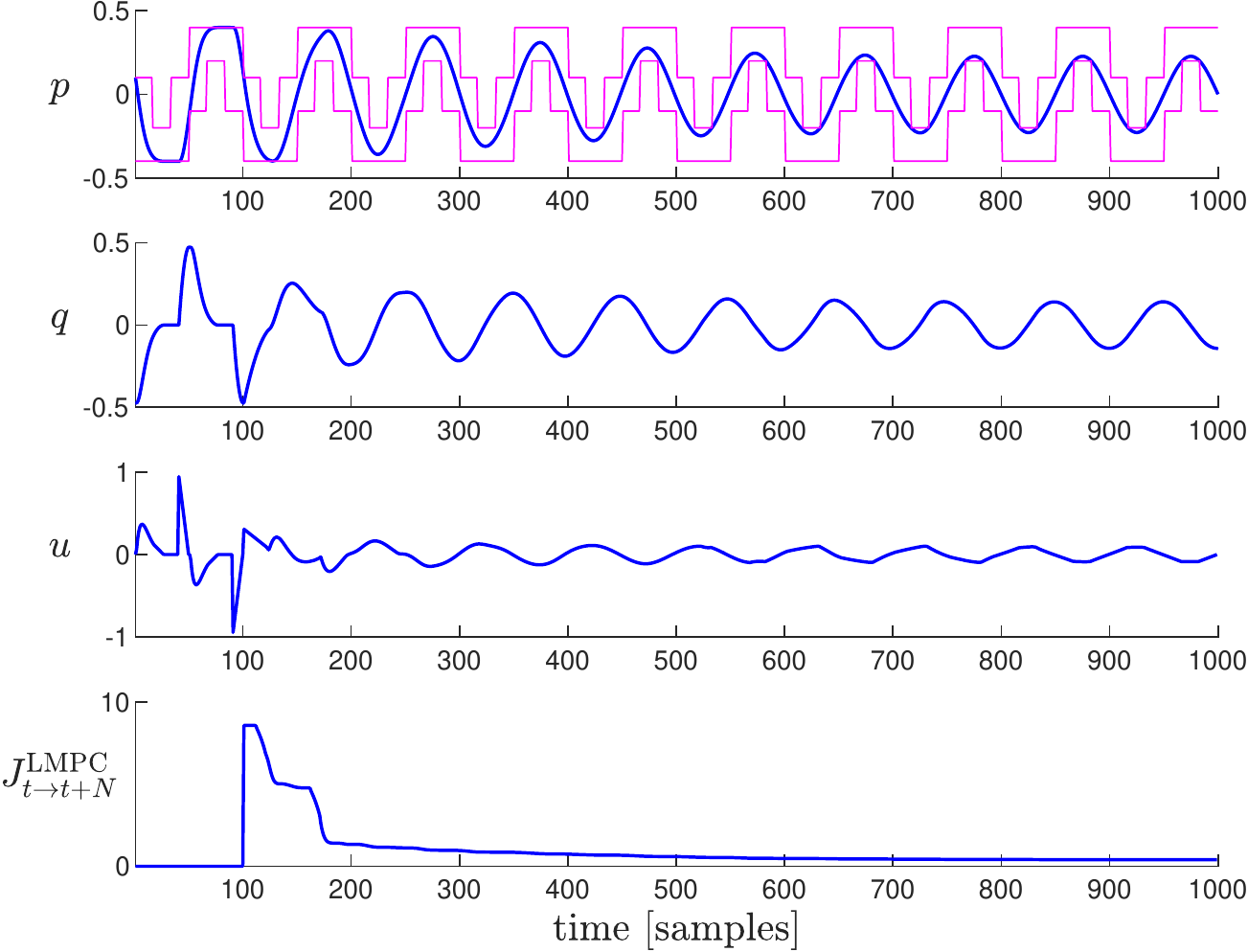}
\caption{Double integrator with time-varying constraints controlled with LMPC.}
\label{fig:TVconstraints}
\end{figure}

The closed-loop trajectory is shown in Figure~\ref{fig:TVconstraints}, where we reported also the time-varying state constraints. Also in this example the state constraint are recursively satisfied.
We notice that the controller is initially following the provided trajectory, and over time attempts to restrict the amplitude of the oscillation. If allowed by the constraints it would regulate the state to the origin, but since this is not possible, it settles on a smooth periodic trajectory touching the inner constraints.

\subsection{Linear system with time-varying stage cost}

In this simulation we use the same double integrator (\ref{eq:doubleIntegrator}), with a velocity constraint $|q|\leq 0.1$, and prediction horizon $N=15$. This time we have a time-varying stage cost, which consists of a set-point on the first component of the state and input minimization. The set-point is switching between $\pm 0.2$ every half period.
\begin{equation}
h_t(x_t,u_t) = 
     \begin{cases}
       (p_t + 0.2)^2 + u_t^2 &\quad\text{if }t \bmod P < P/2\\
       (p_t - 0.2)^2 + u_t^2 &\quad\text{if }t \bmod P \geq P/2 \\ 
     \end{cases}.
\end{equation}
The initial trajectory is a zero input steady-state at the origin. Figure~\ref{fig:TVcost} shows the resulting behavior. The initial cycles attempt to get closer to the set-point but the oscillation is limited by the necessity to satisfy the terminal constraint. As more points are added to the safe set, the system is able to oscillate until reaching the set point. The velocity constraints limit how fast the set point can be reached and the system settles on a periodic trajectory within 4-5 cycles.
\begin{figure}[h!]
\centering
\includegraphics[width=\columnwidth]{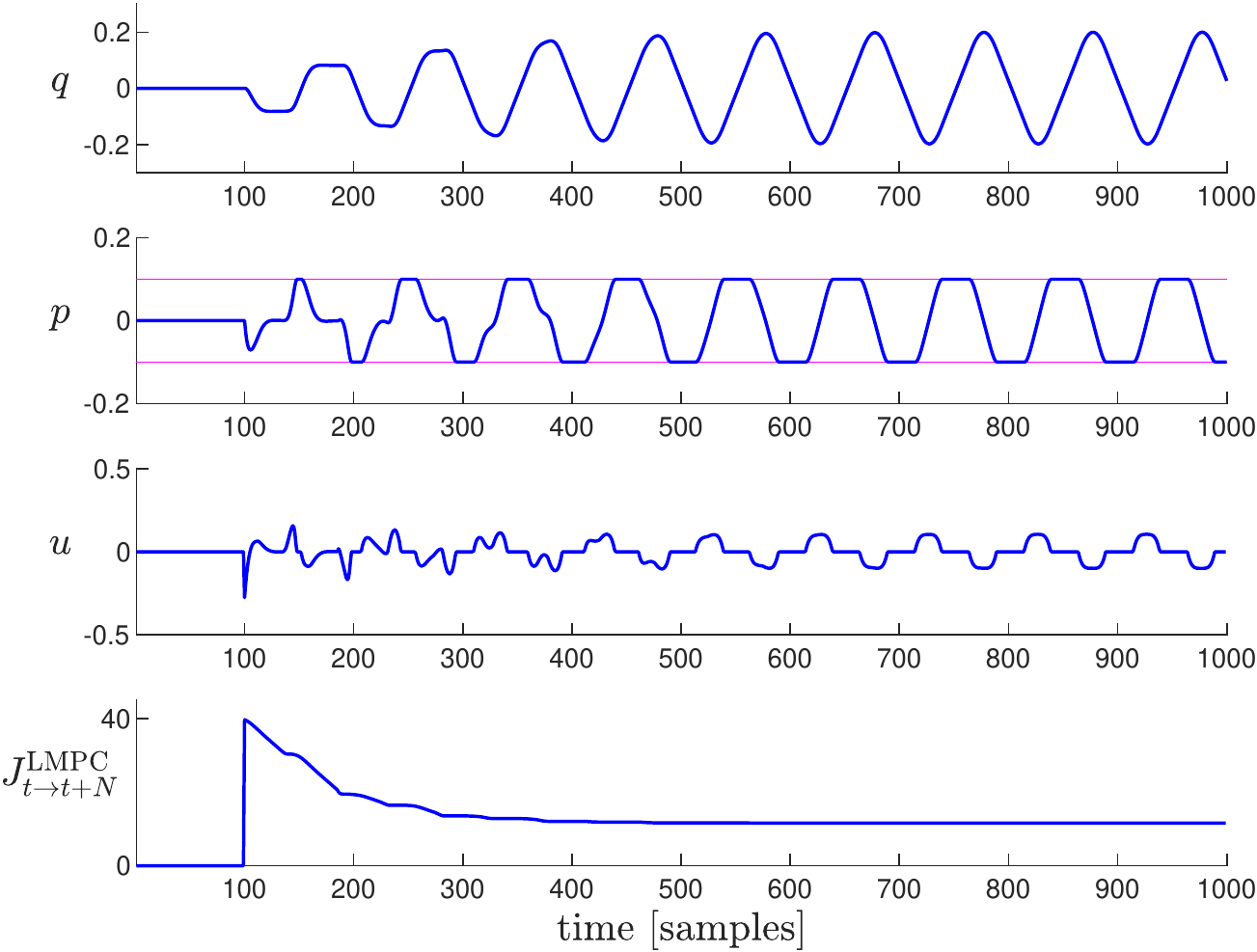}
\caption{Double integrator with time-varying stage cost controlled with LMPC.}
\label{fig:TVcost}
\end{figure}

\subsection{Non-linear time-varying system}

The last simulation is performed on a non-linear system with a time-varying component. The state is $x=(p,q)^T$ and the dynamics are
\begin{equation}
\begin{split}\label{eq:NLsystem}
\left(\begin{array}{c}
p_{t+1} \\
q_{t+1}
\end{array}\right) 
= 
\left(\begin{array}{c}
p_t \\
q_t
\end{array}\right)
+0.1
\left(\begin{array}{c}
q_t \\
p_t (5\sin (2\pi t/P) + u_t)
\end{array}\right),
\end{split}
\end{equation}
with $P=100$, and the control objective is to minimize the following cost $h_t(x_t,u_t) = (p_t-2)^2 \quad \forall t\ge 0$
subject to constraints
$p_t\ge 0.5 \quad |u_t|\le 5 \quad \forall t\ge 0$.

This system satisfies Assumption \ref{eq:dynamicsAssumption}. To verify this consider a collection of states $\{x_j=(p_j,q_j)^T\}$ for $j=0,\dots,L$. The set of multipliers $\{\gamma_j \}$ that verifies the assumption is 
\begin{equation*}
\gamma_j = \frac{\lambda_j p_j}{\sum_{i=0}^L\lambda_i p_i},
\end{equation*}
which if $p_j>0$ gives $0\ge \gamma_j \ge 1$ and $\sum_j\gamma_j = 1$. Also notice that since the stage cost is not a function of the input, this satisfies Assumption \ref{eq:costAssumption}.

The initial trajectory is simply computed by setting $x_t=x_{t+1}$ in (\ref{eq:NLsystem}), which gives $u_t = -5\sin (2\pi t/P)$. It is easy to verify that by picking $x_0=(1,0)^T$ and iterating the dynamics for $0\le t \le P-1$, the constraints remain verified, so the trajectory is feasible.

Figure~\ref{fig:NL} shows a simulation with $N=8$, performed in MATLAB using IPOPT. As for the other simulations, the first $100$ samples show the initial feasible trajectory, which is why they have no associated LMPC cost. From $t=P$ on, the cost is always decreasing, and the system converges to an optimal trajectory. The trajectory actually reaches the set-point $p=2$ because this is the only control objective, but it is unable to stay there when the time-varying term becomes predominant, and thus settles on a periodic trajectory.
\begin{figure}[h!]
\centering
\includegraphics[width=\columnwidth]{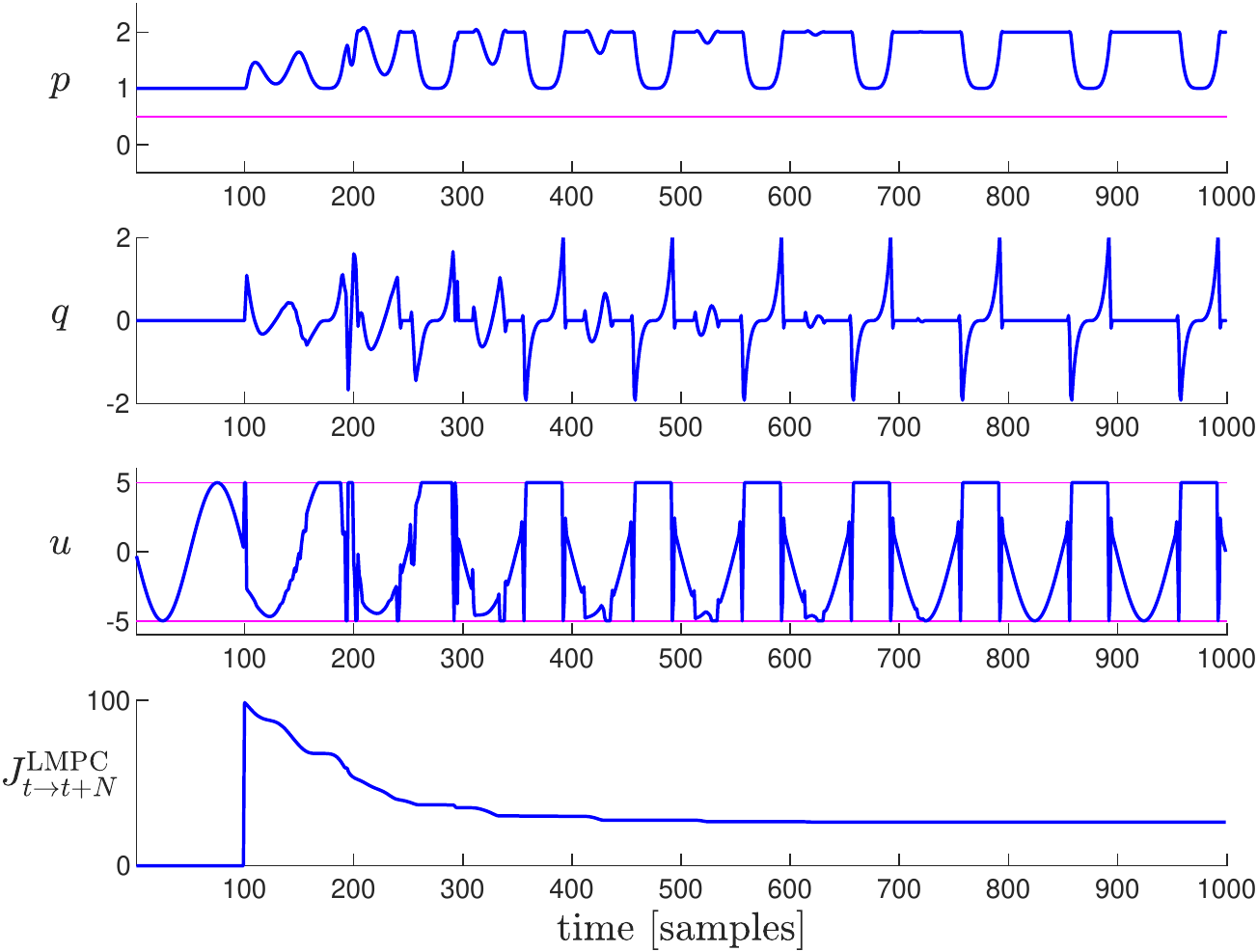}
\caption{Nonlinear time-varying system controlled with LMPC.}
\label{fig:NL}
\end{figure}

\section{Conclusions}

We presented an LMPC for periodic repetitive tasks, considering a wide range of systems defined by dynamics, constraints and stage cost which are periodically time-varying. The controller is aimed at continuous operation and uses historical data to construct a time-varying terminal set and associate to each state a terminal cost. We proved the proposed strategy guarantee recursive constraint satisfaction, and non-increasing MPC cost. Simulations on linear and nonlinear systems prove the effectiveness of the technique.

\renewcommand{\baselinestretch}{0.95}
\bibliographystyle{IEEEtran}
\bibliography{biblio}

\begin{thebibliography}{10}
\providecommand{\url}[1]{#1}
\csname url@samestyle\endcsname
\providecommand{\newblock}{\relax}
\providecommand{\bibinfo}[2]{#2}
\providecommand{\BIBentrySTDinterwordspacing}{\spaceskip=0pt\relax}
\providecommand{\BIBentryALTinterwordstretchfactor}{4}
\providecommand{\BIBentryALTinterwordspacing}{\spaceskip=\fontdimen2\font plus
\BIBentryALTinterwordstretchfactor\fontdimen3\font minus
  \fontdimen4\font\relax}
\providecommand{\BIBforeignlanguage}[2]{{%
\expandafter\ifx\csname l@#1\endcsname\relax
\typeout{** WARNING: IEEEtran.bst: No hyphenation pattern has been}%
\typeout{** loaded for the language `#1'. Using the pattern for}%
\typeout{** the default language instead.}%
\else
\language=\csname l@#1\endcsname
\fi
#2}}
\providecommand{\BIBdecl}{\relax}
\BIBdecl

\bibitem{bristow2006survey}
D.~A. Bristow, M.~Tharayil, and A.~G. Alleyne, ``A survey of iterative learning
  control,'' \emph{IEEE control systems magazine}, vol.~26, no.~3, pp. 96--114,
  2006.

\bibitem{c7}
J.~H. Lee and K.~S. Lee, ``Iterative learning control applied to batch
  processes: An overview,'' \emph{Control Engineering Practice}, vol.~15,
  no.~10, pp. 1306--1318, 2007.

\bibitem{c8}
Y.~Wang, F.~Gao, and F.~J. Doyle, ``Survey on iterative learning control,
  repetitive control, and run-to-run control,'' \emph{Journal of Process
  Control}, vol.~19, no.~10, pp. 1589--1600, 2009.

\bibitem{hillerstrom1996repetitive}
G.~Hillerstr{\"o}m and K.~Walgama, ``Repetitive control theory and
  applications-a survey,'' \emph{IFAC Proceedings Volumes}, vol.~29, no.~1, pp.
  1446--1451, 1996.

\bibitem{lee2001model}
J.~H. Lee, S.~Natarajan, and K.~S. Lee, ``A model-based predictive control
  approach to repetitive control of continuous processes with periodic
  operations,'' \emph{Journal of Process Control}, vol.~11, no.~2, pp.
  195--207, 2001.

\bibitem{gupta2006period}
M.~Gupta and J.~H. Lee, ``Period-robust repetitive model predictive control,''
  \emph{Journal of Process Control}, vol.~16, no.~6, pp. 545--555, 2006.

\bibitem{gondhalekar2011mpc}
R.~Gondhalekar and C.~N. Jones, ``Mpc of constrained discrete-time linear
  periodic systems - a framework for asynchronous control: Strong feasibility,
  stability and optimality via periodic invariance,'' \emph{Automatica},
  vol.~47, no.~2, pp. 326--333, 2011.

\bibitem{cao2008repetitive}
R.~Cao and K.-S. Low, ``A repetitive model predictive control approach for
  precision tracking of a linear motion system,'' \emph{IEEE Transactions on
  Industrial Electronics}, vol.~56, no.~6, pp. 1955--1962, 2008.

\bibitem{balaji2007repetitive}
S.~Balaji, A.~Fuxman, S.~Lakshminarayanan, J.~Forbes, and R.~Hayes,
  ``Repetitive model predictive control of a reverse flow reactor,''
  \emph{Chemical Engineering Science}, vol.~62, no.~8, pp. 2154--2167, 2007.

\bibitem{friis2011repetitive}
J.~Friis, E.~Nielsen, J.~Bonding, F.~D. Adegas, J.~Stoustrup, and P.~F.
  Odgaard, ``Repetitive model predictive approach to individual pitch control
  of wind turbines,'' in \emph{2011 50th IEEE Conference on Decision and
  Control and European Control Conference}.\hskip 1em plus 0.5em minus
  0.4em\relax IEEE, 2011, pp. 3664--3670.

\bibitem{gilbert1977optimal}
E.~G. Gilbert, ``Optimal periodic control: A general theory of necessary
  conditions,'' \emph{SIAM Journal on Control and Optimization}, vol.~15,
  no.~5, pp. 717--746, 1977.

\bibitem{bittanti1986optimal}
S.~Bittanti and G.~Guardabassi, ``Optimal periodic control and periodic systems
  analysis: An overview,'' in \emph{1986 25th IEEE Conference on Decision and
  Control}.\hskip 1em plus 0.5em minus 0.4em\relax IEEE, 1986, pp. 1417--1423.

\bibitem{colonius2006optimal}
F.~Colonius, \emph{Optimal periodic control}.\hskip 1em plus 0.5em minus
  0.4em\relax Springer, 2006, vol. 1313.

\bibitem{ellis2014tutorial}
M.~Ellis, H.~Durand, and P.~D. Christofides, ``A tutorial review of economic
  model predictive control methods,'' \emph{Journal of Process Control},
  vol.~24, no.~8, pp. 1156--1178, 2014.

\bibitem{diehl2010lyapunov}
M.~Diehl, R.~Amrit, and J.~B. Rawlings, ``A lyapunov function for economic
  optimizing model predictive control,'' \emph{IEEE Transactions on Automatic
  Control}, vol.~56, no.~3, pp. 703--707, 2010.

\bibitem{zanon2016periodic}
M.~Zanon, L.~Gr{\"u}ne, and M.~Diehl, ``Periodic optimal control, dissipativity
  and mpc,'' \emph{IEEE Transactions on Automatic Control}, vol.~62, no.~6, pp.
  2943--2949, 2016.

\bibitem{rosolia2017learning}
U.~Rosolia and F.~Borrelli, ``Learning model predictive control for iterative
  tasks. a data-driven control framework,'' \emph{IEEE Transactions on
  Automatic Control}, vol.~63, no.~7, pp. 1883--1896, 2017.

\bibitem{rosolia2017linear}
------, ``Learning model predictive control for iterative tasks: A
  computationally efficient approach for linear system,''
  \emph{IFAC-PapersOnLine}, vol.~50, no.~1, pp. 3142--3147, 2017.

\end{thebibliography}

\end{document}